\documentclass[11pt,a4paper]{article}
\usepackage[lmargin=1.0in,rmargin=1.0in,bottom=1.0in,top=1.0in,twoside=False]{geometry}

\usepackage{authblk}

\usepackage{amssymb,amsmath}
\usepackage{graphicx}
\usepackage{enumerate}
\usepackage{tikz}
\usetikzlibrary{shapes}

\usepackage[T1]{fontenc}

 \usepackage{xcolor}
 \usepackage{mathtools}
\usepackage{microtype}
\usepackage{amsfonts}
\usepackage{comment}
\usepackage{mathrsfs}
\usepackage[vlined, ruled, linesnumbered]{algorithm2e}
\DontPrintSemicolon

\usepackage{diagbox}

\usepackage{trimspaces}
\usepackage{nccfoots}
\usepackage{setspace}
\usepackage{inconsolata}
\usepackage{libertine}
\usepackage[absolute]{textpos}
\usepackage{thmtools}
\usepackage{thm-restate}

\usepackage{todonotes}

\usepackage{longtable}

\definecolor{blue}{rgb}{0.1,0.2,0.5}
\definecolor{brown}{rgb}{0.6,0.6,0.2}
\usepackage[ocgcolorlinks, linkcolor={blue}, citecolor={brown}]{hyperref}

\usepackage{comment}

\usepackage[amsmath,thmmarks,hyperref]{ntheorem}
\usepackage{cleveref}

\crefformat{page}{#2page~#1#3}%
\Crefformat{page}{#2Page~#1#3}%
\crefformat{equation}{#2(#1)#3}%
\Crefformat{equation}{#2(#1)#3}%
\crefformat{figure}{#2Figure~#1#3}%
\Crefformat{figure}{#2Figure~#1#3}%
\crefformat{section}{#2Section~#1#3}
\Crefformat{section}{#2Section~#1#3}
\crefformat{chapter}{#2Chapter~#1#3}
\Crefformat{chapter}{#2Chapter~#1#3}
\crefformat{chapter*}{#2Chapter~#1#3}
\Crefformat{chapter*}{#2Chapter~#1#3}
\crefformat{part}{#2Part~#1#3}
\Crefformat{part}{#2Part~#1#3}
\crefformat{enumi}{#2(#1)#3}
\Crefformat{enumi}{#2(#1)#3}

\usepackage{enumerate}

\usepackage{latexsym}

\theoremnumbering{arabic}
\theoremstyle{plain}
\theoremsymbol{}
\theorembodyfont{\itshape}
\theoremheaderfont{\normalfont\bfseries}
\theoremseparator{.}

\newtheorem{theorem}{Theorem}
\crefformat{theorem}{#2Theorem~#1#3}
\Crefformat{theorem}{#2Theorem~#1#3}

\newcommand{\newtheoremwithcrefformat}[2]{%
  \newtheorem{#1}[theorem]{#2}%
  \crefformat{#1}{##2\MakeUppercase#1~##1##3}%
  \Crefformat{#1}{##2\MakeUppercase#1~##1##3}%
}
\newcommand{\newseptheoremwithcrefformat}[2]{%
  \newtheorem{#1}{#2}%
  \crefformat{#1}{##2\MakeUppercase#1~##1##3}%
  \Crefformat{#1}{##2\MakeUppercase#1~##1##3}%
}

\newtheoremwithcrefformat{proposition}{Proposition}
\newtheoremwithcrefformat{observation}{Observation}
\newtheoremwithcrefformat{conjecture}{Conjecture}
\newtheoremwithcrefformat{corollary}{Corollary}
\newseptheoremwithcrefformat{claim}{Claim}
\theorembodyfont{\upshape}
\newtheoremwithcrefformat{example}{Example}
\newtheoremwithcrefformat{remark}{Remark}
\newseptheoremwithcrefformat{definition}{Definition}
\newseptheoremwithcrefformat{question}{Question}

\crefname{theorem}{Theorem}{Theorems}

\theoremstyle{nonumberplain}
\theoremheaderfont{\scshape}
\theorembodyfont{\normalfont}
\theoremsymbol{\ensuremath{\square}}
\newtheorem{proof}{Proof}

\theoremsymbol{\ensuremath{\lrcorner}}

\def\cqedsymbol{\ifmmode$\lrcorner$\else{\unskip\nobreak\hfil
\penalty50\hskip1em\null\nobreak\hfil$\lrcorner$
\parfillskip=0pt\finalhyphendemerits=0\endgraf}\fi}

\tikzset{
    position/.style args={#1:#2 from #3}{
        at=(#3.#1), anchor=#1+180, shift=(#1:#2)
    }
}

\newcommand{\Oh}{\mathcal{O}}

\let\originalleft\left
\let\originalright\right
\renewcommand{\left}{\mathopen{}\mathclose\bgroup\originalleft}
\renewcommand{\right}{\aftergroup\egroup\originalright}

\renewcommand{\leq}{\leqslant}
\renewcommand{\geq}{\geqslant}

\usepackage{todonotes}

\definecolor{lars}{HTML}{0F3CCB}

\definecolor{larsA}{HTML}{196F3D}
\definecolor{larsB}{HTML}{FF8000}
\definecolor{larsC}{HTML}{DF0101}

\usepackage{lineno}

\usepackage{macros}

\newcommand{\frontpageformat}{arxiv}

\begin{document}

\ifthenelse{\equal{\frontpageformat}{submission}}{%
\author{anonymous}
\title{A note on independent sets in~sparse-dense~graphs}
\begin{titlepage}
\def\thepage{}
\thispagestyle{empty}
\maketitle
}{%
\author[1,2]{Ueverton S.\ Souza}

\affil[1]{Fluminense Federal University, Brazil}
\affil[2]{University of Warsaw, Poland}

\title{A note on independent sets in~sparse-dense~graphs%
\thanks{
This research has received funding from Rio de Janeiro Research Support Foundation (FAPERJ) under grant agreement E-26/201.344/2021,  National Council for Scientific and Technological Development (CNPq) under grant agreement 309832/2020-9, and the European Research Council (ERC) under the European Union's Horizon 2020 research and innovation programme under grant agreement CUTACOMBS (No. 714704).
}}

\begin{titlepage}
\def\thepage{}
\thispagestyle{empty}
\maketitle
\begin{textblock}{20}(0, 13.3)
\includegraphics[width=40px]{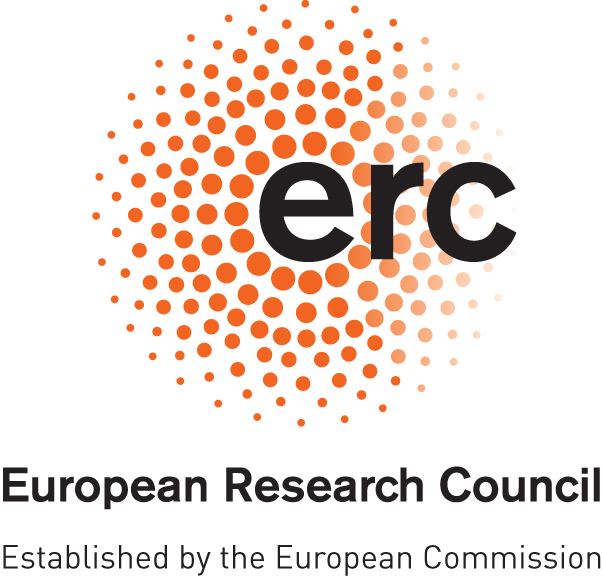}%
\end{textblock}
\begin{textblock}{20}(0, 14.1)
\includegraphics[width=40px]{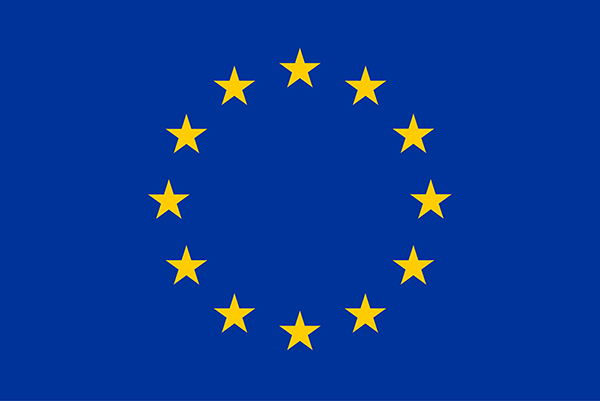}%
\end{textblock}
}

\begin{abstract}
The notion of sparse-dense partitions of graphs was introduced by Feder, Hell, Klein, and Motwani [STOC 1999, SIDMA 2003] as a tool to solve several partitioning problems. A sparse-dense partition of a graph $G$ is a partition of $V(G)$ into two parts $V(G) = S \cup D$ where $S \in \mathcal{S}$ and $D \in \mathcal{D}$ for some graph classes $\mathcal S$ and $\mathcal D$ such that: both $\mathcal S$ and $\mathcal D$ are closed under induced subgraphs; and there is a constant $c$ so that for any $S' \in \mathcal{S}$ and $D' \in \mathcal{D}$ the intersection $S' \cap D'$ has at most $c$ vertices. Graphs admitting sparse-dense partitions for some classes $\mathcal S$ and $\mathcal D$ are called sparse-dense graphs. 
This paper presents some results concerning maximal independent sets in sparse-dense graphs. It is proved~that if a graph $G$ admits a sparse-dense partition concerning classes $\mathcal S$ and $\mathcal D$, where $\mathcal D$ is a subclass of the complement of $K_t$-free graphs (for some ~$t$), and the graph class $\mathcal S$ can be recognized in polynomial time, then:
\begin{itemize}
    \item enumerating all maximal independent sets (or finding a maximum independent set) of $G$ can be performed in polynomial time whenever it can be done in polynomial time for graphs in the class $\mathcal S$.
\end{itemize}
This result has the following interesting implications:
\begin{enumerate}
    \item A \Poly~versus \np-hard dichotomy for {\sc Max.~Independent Set} on $(k, \ell)$-graphs concerning the values of $k$ and $\ell$. The class of $(k, \ell)$-graphs is the class of graphs $G$ where $V(G)$ can be partitioned into $k$ independent sets and $\ell$ cliques.
    \item 
    A \Poly-time algorithm that does not require $(1,\ell)$-partitions for determining whether a $(1,\ell)$-graph $G$ is well-covered.  Well-covered graphs are graphs in which every maximal independent set has the same cardinality.
    \item The characterization of conflict graph classes for which the conflict version of a {\Poly-time} graph problem is still in {\Poly} assuming such classes. Conflict versions of graph problems ask for solutions avoiding pairs of conflicting elements (vertices/edges) described in conflict graphs. 
\end{enumerate}    
\end{abstract}

\end{titlepage}

%
%
%
%
%

\section{Introduction}

Partitioning graph problems consist in separating the vertices of a given graph $G$ into $k$ parts $V_1,\ldots, V_k$ such that $V_1 \cup \ldots \cup V_k = V(G)$, $V_i \cap V_j = \emptyset$ ($1 \leq i,j \leq k$ and $i \neq j$), and each part $V_i$ satisfies a property $\Pi_i$.  Sometimes it can be also required that the edge set between $V_i$ and $V_j$ satisfies a property $\Pi_{i,j}$ for $1 \leq i,j \leq k$ and $i \neq j$.  Usually, some of these internal ($\Pi_i$) or external ($\Pi_{i,j}$) properties are irrelevant, allowing to be any set of vertices or edges (including the empty set). Classical partitioning graph problems are {\sc $k$-Coloring} and {\sc $k$-Clique Cover} which aim partition the vertices of a given graph $G$ into $k$ independent sets or $k$ cliques, respectively.   Also,  Feder, Hell, Klein, and Motwani~\cite{FederHKM03,feder1999complexity} described a family of partitioning problems called $M$-partition problems where the ``pattern of requirements'' is given by a fixed symmetric $k$-by-$k$ matrix.
The recognition of \textit{split graphs} is another example of a partitioning problem. A graph $G$ is split if it admits a partition of its vertex set into a clique and an independent set. 
A generalization of split graphs defined by Brandst{\"a}dt~\cite{brandstadt1996partitions} is the  notion of $(k, \ell)$\textit{-graphs}, that is graphs whose vertex set admits a partition into $k$ independent sets and $\ell$ cliques.  Such partitions are called $(k, \ell)$-\emph{partitions}. For convenience, some parts are allowed to be empty, and a single vertex set can be counted as either an independent set or a clique. A {\Poly} versus \np-complete dichotomy for recognizing $(k, \ell)$-graphs was proved by Brandst{\"a}dt~\cite{brandstadt1996partitions}: the problem is in  {\Poly} if $\max\{k, \ell\} \leq 2$, and is $\NP$-complete otherwise.
Any $(k, \ell)$-partition can be seen as a partition of $V(G)$ into a ``sparse'' graph with chromatic number at most $k$ and a ``dense'' graph with clique cover number at most $\ell$. This notion of partitioning graphs into ``sparse'' and ``dense'' parts was generalized by Feder, Hell, Klein, and Motwani~\cite{FederHKM03,feder1999complexity} as a tool to solve several partitioning problems.

Let $\mathcal S$ and $\mathcal D$ be two classes of graphs. Feder, Hell, Klein, and Motwani~\cite{FederHKM03,feder1999complexity} described $\mathcal S$ as a class of sparse graphs and $\mathcal D$ as a class of dense graphs if $\mathcal S$ and $\mathcal D$ satisfy the following constraints:
\begin{itemize}
    \item both $\mathcal S$ and $\mathcal D$ are closed under induced subgraphs;
    \item there exists a constant $c$ such that the intersection ${S} \cap {D}$ has at most $c$ vertices for any $S \in \mathcal{S}$ and $D \in \mathcal{D}$.
\end{itemize}

A \emph{sparse-dense partition} of a graph $G$, with respect to the classes $\mathcal S$ and $\mathcal D$ of ``sparse'' and ``dense'' graphs, is a partition of $V(G)$ into two parts $V(G) = S \cup D$ such that $S \in \mathcal{S}$ ($S$ is sparse) and $D \in \mathcal{D}$ ($D$ is dense). Graphs admitting sparse-dense partitions for some classes $\mathcal S$ and $\mathcal D$ are called \emph{sparse-dense graphs}. 
Some examples of sparse-dense partitions are partitions into a graph having bounded clique size and a graph having bounded independence number; partitions into a planar graph and a clique; or partitions into a bipartite graph and a co-bipartite graph. For partitioning problems modeled by sparse-dense partitions, one can apply the following theorem due to Feder, Hell, Klein, and Motwani:

\begin{theorem}[Sparse-Dense Theorem~\cite{FederHKM03,feder1999complexity}]
Let $\mathcal S$ and $\mathcal D$ be classes of sparse and dense graphs, respectively. A graph on $n$ vertices has at most $n^{2c}$ different sparse-dense partitions. Furthermore, all these partitions can be found in time proportional to $n^{2c+2}\cdot T(n)$, where $T(n)$ is the time for recognizing the sparse and dense graphs.
\end{theorem}

The Sparse-Dense Theorem implies that partitioning problems modeled by sparse-dense partitions can be solved in polynomial time whenever both classes $\mathcal S$ and $\mathcal D$ are polynomial-time recognizable. 

This paper concerns the \textsc{Maximum Independent Set} problem on sparse-dense graphs. It is proved~that if a graph $G$ admits a sparse-dense partition concerning classes $\mathcal S$ and $\mathcal D$ where graphs in $\mathcal S$ can be recognized in polynomial time, and $\mathcal D$ is a subclass of the complement of $K_t$-free graphs (for some nonnegative integer $t$) then:

\begin{itemize}
    \item enumerating all maximal independent sets (or finding a maximum independent set) of $G$ can be performed in polynomial time whenever it can be done in polynomial time for graphs in $\mathcal S$.
\end{itemize}

Some examples of sparse classes $\mathcal S$ satisfying the desired constraints are empty graphs (independent sets), trees, and complete $k$-partite graphs for fixed~$k$. Also, by considering finding a maximum independent set instead of enumerating maximal independent sets, one can handle more general graph classes $\mathcal S$ such as bipartite and bounded treewidth graphs. In addition, examples of dense classes $\mathcal D$ are complete graphs, co-bipartite graphs, co-planar graphs, and graphs having bounded co-degeneracy~\cite{duarteMFCS}. 

As a by-product of the main result, 
contributions in the regime of $(k,\ell)$-graphs, well-covered graphs, and conflict-free graph problems are also achieved. These results reveal an applicable structural property for determining tractable instances of problems related to independent sets.

\subsection*{On $(k,\ell)$-graphs}

The problem of recognizing subclasses of $(k,\ell)$-graphs and the complexity of problems on $(k,\ell)$-graphs have been extensively studied~\cite{ITOR_coloracao,baste2017parameterized,DBLP:journals/disopt/DemangeEW05,rlmaxcut,FederHKM03,FederHKNP05,KolayPRS15}.  Sometimes $(k,\ell)$-graphs are also called $(r,\ell)$-graphs due to the common use of $k$ to denote parameters in parameterized complexity.

From the results of this paper, a {\Poly} vs. \-NP-hard dichotomy of \textsc{Maximum Independent Set} regarding to the parameters $k$ and $\ell$ of $(k, \ell)$-graphs is achieved. This determines the boundaries of the {\np}-hardness for such graph classes. The main results of this paper regarding $(k, \ell)$-graphs are:
\begin{enumerate}
    \item the problem of enumerating all maximal independent sets of a graph can be done in polynomial time on the class of $(1,\ell)$-graphs for any $\ell\geq 0$;
    \item the problem of finding a maximum independent set of a graph can be done in polynomial time on the class of $(2,\ell)$-graphs for any $\ell\geq 0$.  
\end{enumerate}

Also, by replacing each edge $e=uv$ of a graph with a path $ux_ey_ev$ where $x_e,y_e$ are new vertices, the problem of finding a maximum independent set is reducible to the case where the input graph is a $(3,0)$-graph. 
Therefore, the complete classification regarding the time complexity
of {\sc Maximum Independent Set} on $(k,\ell)$-graphs is established (see Table 1).


\begin{table}[h!tb]
\centering
\begin{tabular}{|c|c|c|}
\hline 
\diagbox{$k$}{$\ell$} & 0 & $\geq 1$ \\ \hline
$0$        &
--    &
\textcolor{teal}{\Poly} \\ \hline

$1$        & 
\textcolor{teal}{\Poly}  & 
\begin{tabular}{c}\textcolor{teal}{\Poly} \\ (Cor.~\ref{teo:1lcase})\end{tabular} \\ \hline

$2$        & 
\textcolor{teal}{\Poly} & 
\begin{tabular}{c}\textcolor{teal}{\Poly} \\ (Cor.~\ref{teo:2lcase})\end{tabular} \\ \hline

$\geq 3$   & 
\begin{tabular}{c}\textcolor{red}{\np-hard} \\ (Claim.~\ref{claim:30})\end{tabular} & 
\begin{tabular}{c}\textcolor{red}{\np-hard} \\ (Claim.~\ref{claim:30})\end{tabular} \\ \hline
\end{tabular}
\caption{{\Poly} vs. \-NP-hard dichotomy of \textsc{Maximum Independent Set} on $(k,\ell)$-graphs}
\label{dicotomia}
\end{table}

\subsection*{On well-covered graphs}

The \emph{Well-covered graph class} is the class of graphs in which every maximal independent set has the same cardinality~\cite{PLUMMER197091}, and recognizing such a class is {\sf coNP}-complete~\cite{chvatal1993note,sankaranarayana1992complexity}. 
Studies regarding structural characterizations of particular subclasses of well-covered graphs has been considered for a long time~\cite{fcf48a9c6bc04e0f8917664daa376300,doi:10.1002/jgt.3190180707,FINBOW199344,DBLP:journals/gc/KleinMM13,doi:10.1002/(SICI)1097-0118(199602)21:2<113::AID-JGT1>3.0.CO;2-U,ravindra1977well,TANKUS1996293}.
In addition,  parameterized algorithms for recognizing well-covered graphs were presented in~\cite{AlvesDFKSS18,DBLP:journals/dmtcs/AraujoCKSS19}. 

Let $k, \ell \geq 0$ be two fixed integers, not simultaneously 
zero. A graph is {\it $(k, \ell)$-well-covered} if it is both $(k, \ell)$ and well-covered. Recently, Alves, Couto, Faria, Gravier, Klein, and Souza~\cite{AlvesFKGSS20} studied the complexity of {\sc {Graph Sandwich}} for the property of being $(k,\ell)$-well-covered, and Faria and Souza~\cite{cocoon21} studied the complexity of {\sc {Probe}} problems for the property of being $(k,\ell)$-well-covered.
In addition, Alves, Dabrowski, Faria, Klein, Sau, and Souza~\cite{AlvesDFKSS18} established the complete classification of the complexity of recognizing $(k,\ell)$-well-covered graphs.
In particular, they show that recognizing if a graph $G$ is $(1,\ell)$-well-covered is \np-complete for each $\ell\geq 3$, but given an $(1,\ell)$-partition of $G$ to determine whether $G$ is well-covered can be done in polynomial time. Since the problem of finding a $(1,\ell)$-partition is \np-hard for $\ell\geq 3$~\cite{brandstadt1996partitions}, given a graph $G$ that is known to be $(1,\ell)$, to determine whether $G$ is well-covered is challenging when no $(1,\ell)$-partition is given together with $G$.
%
However, since $(0,\ell)$-graphs are $\overline{K_{\ell+1}}$-free, from the results of this paper, it holds that determining whether a given $(1,\ell)$-graph $G$ is well-covered can also be done in polynomial time even when no $(1,\ell)$-partition is provided in the input.

\subsection*{On conflict-free graph problems}

Conflict variants of classical computational problems ask to avoid conflicting elements in problem solutions.  This constraint is natural in real-world applications.  A conflict version of a problem $\mathcal{X}$ is obtained by attaching conflict graphs $\hat{G}$ together with the instances $I_\mathcal{X}$ of $\mathcal{X}$.  In such a  conflict graph, the vertices represent elements of $I_\mathcal{X}$, and its edges represent pairs of elements that are forbidden to be mutually in the same solution. A solution for an instance $(I_\mathcal{X}, \hat{G})$ of a conflict version of $\mathcal{X}$ represents an independent set in $\hat{G}$ and a solution for the instance $I_\mathcal{X}$~of~$\mathcal{X}$.
Conflict versions of {\sc Bin Packing}, {\sc Knapsack},  {\sc Maximum Matching},  {\sc Shortest Path} and {\sc Spanning Tree} have been considered in the literature~\cite{capua2018study,CALDAM2022,darmann2011,gendreau2004heuristics,pferschy2009knapsack,zhang}.

In~\cite{darmann2011}, Darmann, Pferschy, Schauer, and Woeginger proved that the conflict version of {\sc Maximum Matching} remains \np-hard even when the conflict graph is an induced matching (i.e., disjoint union of $K_2$'s). This exemplifies the difficulty of finding nice properties on conflict graphs so that problems originally solvable in polynomial time are still tractable in their version having conflict graphs satisfying these particular properties. Contrastingly, as a corollary of the results of this paper, it follows that if $\mathcal{X}$ is a graph problem in {\Poly} then $(I_\mathcal{X}, \hat{G})$ can be solved in polynomial time whenever $\hat{G}$ is a sparse-dense graph with respect to classes $\mathcal S$ and $\mathcal D$, where $\mathcal D$ is a subclass of $\overline{K_t}$-free graphs (for some nonnegative integer $t$) and $\mathcal S$ is a polynomial-time recognizable graph class such that its maximal independent sets can be enumerated in polynomial time. Examples of such conflict graphs $\hat{G}$ include $(1,\ell)$-graphs.

\section{Results}


\begin{theorem}\label{teo:enumerate}
Let $t$ be a nonnegative integer, and $G$ be a sparse-dense graph concerning $\mathcal S$ and $\mathcal D$ such that $\mathcal D$ is a subclass of $\overline{K_t}$-free graphs and $\mathcal S$ is a graph class recognizable in $T_1(n)$ time whose maximal independent sets can be enumerated in $T_2(n)$ time on $n$-vertex graphs. All maximal independent sets of $G$ can be enumerated in $(T_1(n)+T_2(n))\cdot n^{\Oh(t)}$ time.
\end{theorem}
\begin{proof}
Recognizing $\overline{K_{t}}$-free graphs can be done in $n^{\Oh(t)}$ time. 
Therefore, by Sparse-Dense Theorem, one can obtain a sparse-dense partition $V(G)=S\cup D$ such that $S\in \mathcal S$ and $D$ is $\overline{K_t}$-free in time $T_1(n)\cdot n^{\Oh(1)}+n^{\Oh(t)}$.
%
Now, one can enumerate all independent sets $R_D$ of $D$ (including the empty~set). These sets contain at most $t$ vertices and, therefore, can be enumerated in $n^{\Oh(t)}$ time. Also, one can enumerate each maximal independent set $R_S$ of $S$ in $T_2(n)$ time. 
After that, for each pair $R_D$, $R_S$, one could obtain the maximal extension of $R_D$ using vertices of each $R_S$ in polynomial time.   
Finally, the resulting independent sets that are maximal form precisely the collection of all maximal independent sets of $G$, since all possible independent sets of $D$ were checked. 
\end{proof}




Remark that the challenge in order to enumerate independent sets in $(1,\ell)$-graphs is that finding a $(1,\ell)$-partition is \np-hard for $\ell\geq 3$. A $(1,\ell)$-partition of a graph is a sparse-dense partition into a (sparse) independent set $S$ and a (dense) $(0,\ell)$-graph $D$. Unfortunately, the Sparse-Dense Theorem cannot be directly applied to find a $(1,\ell)$-partition as $T(n)$, in this case, would be non-polynomial unless \Poly=\np. 
However, every $(1,\ell)$-partition of a graph $G$ can be seen as a partition of $V(G)$ into an independent set $S$ and a $\overline{K_{\ell+1}}$-free graph $D$ (any $(0,\ell)$-graph is $\overline{K_{\ell+1}}$-free).  Similar behavior occurs with several other partitioning problems whose recognition of the class of graphs that admit such partitions is \np-complete. Therefore, the following holds.

\begin{corollary}\label{cor:classedegrafos}
Let $\mathcal{C_{S,D}}$ be the class of sparse-dense graph concerning $\mathcal S$ and $\mathcal D$ such that $\mathcal D$ is a subclass of $\overline{K_t}$-free graphs and $\mathcal S$ is a graph class recognizable in polynomial time whose maximal independent sets can be enumerated in polynomial time.  There is a polynomial-time algorithm that given a graph $G$ either asserts that $G\notin \mathcal{C_{S,D}}$ or enumerates all maximal independent sets of $G$ in polynomial time.
\end{corollary}

\begin{corollary}\label{teo:1lcase}
For any integer $\ell\geq 0$, enumerating all maximal independent sets can be done in polynomial time for graphs in the class of $(1,\ell)$-graphs.  
\end{corollary}

\begin{corollary}
For any integer $\ell\geq 0$, determining whether a given $(1,\ell)$-graph $G$ is well-covered can be done in polynomial time even when no $(1,\ell)$-partition of $G$ is provided.
\end{corollary}

\begin{corollary}
Let $\mathcal{X}$ be a graph problem in {\Poly}. The conflict-free variant of $\mathcal{X}$ can be solved in polynomial time whenever the conflict graph $\hat{G}$ is a sparse-dense graph with respect to classes $\mathcal S$ and $\mathcal D$, where $\mathcal D$ is a subclass of $\overline{K_t}$-free graphs (for some nonnegative integer $t$) and $\mathcal S$ is polynomial-time recognizable graph class such that its maximal independent sets can be enumerated in polynomial time. In particular, it holds for conflict graphs that are $(1,\ell)$-graphs for any fixed $\ell$.
\end{corollary}
\begin{proof}
By Theorem~\ref{teo:enumerate}, all maximal independent sets of $\hat{G}$ can be enumerated in polynomial time. For each maximal independent set $C$ of $\hat{G}$ all elements (vertices/edges) that conflict with elements in $C$ can be removed from the input graph $G$ (instance of the original problem), and a solution to the problem $\mathcal{X}$ in this new instance (if any) can be found using a polynomial-time algorithm for $\mathcal{X}$. Thus, a feasible solution (if any) can be found in polynomial time for such instances of the conflict version of $\mathcal{X}$, and since any evaluation function of $\mathcal{X}$ (if any) is polynomial-time computable, then the optimal feasible solution with respect to these evaluation functions can also be computed in polynomial time.
\end{proof}


It is easy to see that there are bipartite graphs with an exponential number of maximal independent sets. For instance, induced matchings with $n$ edges have $2^n$ maximal independent sets. Therefore, for sparse-dense graphs like $(2,\ell)$-graphs, the focus is on finding a maximum independent set.

\begin{theorem}
Let $t$ be a nonnegative integer, and $G$ be a sparse-dense graph concerning $\mathcal S$ and $\mathcal D$ such that $\mathcal D$ is a subclass of $\overline{K_t}$-free graphs and $\mathcal S$ is a graph class recognizable in $T_1(n)$ time whose maximum independent set can be found in $T_2(n)$ time on $n$-vertex graphs. A maximum independent set of $G$ can be enumerated in $(T_1(n)+T_2(n))\cdot n^{\Oh(t)}$ time.
\end{theorem}
\begin{proof}
By the Sparse-Dense Theorem one can obtain a sparse-dense partition $V(G)=S\cup D$ such that $S\in \mathcal S$ and $D$ is $\overline{K_t}$-free in $(T_1(n)\cdot n^{\Oh(1)}+n^{\Oh(t)})$ time. By enumerating all possible independent sets $R_D$ of $D$ (including the empty set), taking the graph induced by the subset of $S$ having no neighbor in $R_D$, and performing a $T_2(n)$-time algorithm for {\sc Maximum Independent Set} on this subgraph of $S$, we obtain all relevant independent sets of $G$, where the maximum independent set of $G$ is the largest of them, since all possible subsets of $D$ were analyzed.
\end{proof}

Note that a corollary similar to Corollary~\ref{cor:classedegrafos} applies for classes $\mathcal S$ whose recognition and maximum independent set computation can be performed in polynomial time.  Also, recall that a maximum independent set of a bipartite graph can be found in polynomial time using maximum matching algorithms (c.f. Kőnig's theorem). Therefore, the following holds.

\begin{corollary}\label{teo:2lcase}
For any integer $\ell\geq 0$, finding a maximum independent set can be done in polynomial time for graphs in the class of $(2,\ell)$-graphs. 
\end{corollary}

By replacing each edge $e=uv$ of a graph with a path $ux_ey_ev$, where $x_e,y_e$ are new vertices, the {\sc Maximum Independent Set} problem is reducible to the case where the input is a $(3,0)$-graph. Since empty parts are allowed on $(k,\ell)$-graphs, the following complete the dichotomy regarding $(k,\ell)$-graphs.

\begin{claim}\label{claim:30}
Given two nonnegative integers $k\geq 3$ and $\ell\geq 0$, the problem of finding a maximum independent set of a $(k,\ell)$-graph is \np-hard. 
\end{claim}

Finally, since the complement of a sparse-dense graph is also sparse-dense, we remark that similar results regarding enumerating maximal cliques or finding a maximum clique can be obtained for the complementary graph classes. 


\end{document}